\documentclass[11pt]{article}
\usepackage{amsmath, amsthm, amscd, amsfonts, amssymb, graphicx, color}
\usepackage[utf8]{inputenc}
\usepackage{hyperref}
\usepackage[titlenumbered,linesnumbered,ruled,noend,algosection]{algorithm2e}

\newtheorem{thm}{Theorem}[section]

\newtheorem{lem}[thm]{Lemma}

\newtheorem{ques}[thm]{Question}

\numberwithin{equation}{section}

\begin{document}
\title{An Improved Lower Bound for General Position Subset Selection}
\author{Ali Gholami Rudi\thanks{
{Department of Electrical and Computer Engineering},
{Bobol Noshirvani University of Technology}, {Babol, Iran}.
Email: {\tt gholamirudi@nit.ac.ir}.}}
\date{}

\maketitle

\begin{abstract}
In the General Position Subset Selection (GPSS) problem, the goal is to find
the largest possible subset of a set of points such that no three of
its members are collinear.
If $s_{\textrm{GPSS}}$ is the size of the optimal solution,
$\sqrt{s_{\textrm{GPSS}}}$ is the current best guarantee
for the size of the solution obtained using a polynomial time algorithm.
In this paper we present an algorithm for GPSS to improve this bound
based on the number of collinear pairs of points.
We experimentally evaluate this and few other GPSS algorithms;
the result of these experiments suggests further opportunities for
obtaining tighter lower bounds for GPSS.
\end{abstract}

\section{Introduction}
A subset of a set of $n$ points in the plain is in general position
if no three of its members are on the same line.  The NP-complete
\textsc{General Position Subset Selection} (GPSS) problem asks for the largest
possible such subset.
This problem, the fame of which is partly due to the
fact that several algorithms in computational geometry assume
that their input points are in general position,
has received relatively little attention in its general setting.
The well-known \textsc{No-Three-In-Line} problem, which is a special
case of GPSS, asks for the maximum number of points, no three of which
are collinear in an $n \times n$ grid.  A lower
bound of $(3/2 - \epsilon) n$ was proved for this problem \cite{hall75} and
it is conjectured that the best lower bound
for large $n$ is $cn$ \cite{guy68}, in which $c$ is $\frac{\pi}{\sqrt{3}}$.
\textsc{No-Three-In-Line} has also been extended to three dimensions \cite{por07}.

Lower bounds for GPSS were proved by Payne and Wood for the
case in which the number of collinear points is bounded \cite{payne13}.
More precisely,
if no more than $\ell \le \sqrt{n}$ of input points are collinear,
they showed that the size of the largest subset of points in general position
is $\Omega ( \sqrt{ \frac{n}{\ln \ell}})$.
More recently Froese et al.\ proved that GPSS
is NP-complete and APX-hard \cite{froese16}.  They
also presented several fixed-parameter tractability
results for this problem, including a kernel of size $15 k^3$.

A problem closely related to GPSS is \textsc{Point Line Cover} (PLC).
The goal in PLC is to find the minimum
number of lines that cover a set of points.  This problem
has been minutely studied and an approximation algorithm with
performance ratio $\log (s_{\textrm{PLC}})$ has been presented for this
problem \cite{kratsch16},
in which $s_{\textrm{PLC}}$ is the size of the optimal solution to PLC.
PLC can be used to prove bounds for GPSS \cite{cao12}:
given that at most two points can be selected from each
line of a line cover,
clearly $s_{\textrm{GPSS}} \le 2 \cdot s_{\textrm{PLC}}$,
in which $s_{\textrm{GPSS}}$ is the size of the optimal solution to GPSS.
Also, since ${s_{\textrm{GPSS}} \choose 2}$
lines are defined for a set of $s_{\textrm{GPSS}}$ points in general
position and since all points outside the optimal solution to GPSS
should be on at least one such line (due to its maximality), we have
$s_{\textrm{PLC}} \le {s_{\textrm{GPSS}}}^2$.

Cao presented a greedy algorithm for GPSS which works as follows \cite{cao12}.
Let $S$ be an empty set initially.
For each point $p$ in the set of input points $P$ in some arbitrary order,
add $p$ to $S$, unless it is on a line formed by the points present in $S$.
It is easy to see that in $S$ no three points can be collinear.
On the other hand, due to its incremental construction, $S$ is
maximal and no point in $P \setminus S$ can be added to $S$.
This algorithm achieves the best known approximation ratio for
GPSS \cite{froese16}.  Since each point in an optimal solution $Q$
outside $S$ cannot be added to $S$, it should be on a line defined
by the points in $S$, and since there are
${|S| \choose 2}$ such lines and on
each of these lines at most two points of $Q$ can appear,
$|Q| \le |S| + 2 {|S| \choose 2}$.
This algorithm, therefore, finds a subset of size at least $\sqrt {s_{\textrm{GPSS}}}$.

In this paper, we try to improve this bound by reformulating the
problem using graphs and finding maximal independent sets in them.
Given a set $P$ of $n$ points,
the algorithm presented in this paper finds a subset in general position
with $\max \{ 2 n^2 \mathbin{/} (\mathrm{coll} (P) + 2n) , \sqrt{s_{\textrm{GPSS}}} \}$
points, in which $\mathrm{coll} (P)$ is the total number of collinear pairs in lines
with at least three points in $P$ (Theorem~\ref{tlow}).
We experimentally evaluate this and three other GPSS algorithms.
Our results show that a modification of the algorithm described
in the previous paragraph experimentally obtains larger sets and
may be used to identify a better lower bound for GPSS.

The paper is organized as follows: in Section~\ref{spre}, we define
the notation used in this paper and in Section~\ref{slow}, we describe
our algorithm.  We start Section~\ref{sgen} with a discussion about
the challenges of generating GPSS test cases and how the test cases
used in this paper were obtained.  We then report the result of
our experiments and finally in Section~\ref{scon} we conclude this
paper.

\section{Preliminaries and Notation}
\label{spre}

Let $P$ be a set of $n$ points in the plane.  Three or
more points of $P$ are \emph{collinear} if there is a line that
contains all of them.
Let $L$ be the set of all maximal collinear subsets of $P$.
For each point $p$ in $P$, let $L (p)$ be the subset of $L$
containing all elements of $L$ that contain $p$.
Also, let $N (p)$ denote the union of
all members of $L (p)$, excluding $p$ itself.
We define $\mathrm{coll} (p)$ as the size of $N (p)$,
$\mathrm{coll} (Q)$ for a subset $Q$ of $P$ as the sum of $\mathrm{coll} (q)$
for every $q$ in $Q$, and
$\overline{\mathrm{coll}}(Q)$ as the average value of $\mathrm{coll} (q)$ for
every point $q$ in $Q$.

A subset $Q$ of $P$ is \emph{noncollinear} if, for every $q$ in $Q$,
no point in $L(q)$ is present in $Q$.
The \emph{collinearity graph} $G$ of a finite set $P$ of points
is the graph that has a vertex for each point in $P$;
in this paper we use the same symbol to represent a point
in $P$ and its corresponding vertex in $G$.
Two vertices $p$ and $q$ are adjacent in $G$
if and only if $p$ is in $N(q)$.
It can be observed that the degree of a vertex $p$ in $G$
equals $\mathrm{coll} (p)$.
Figure~\ref{coll} demonstrates these definitions in a
small example.

\begin{figure}
	\centering
	\includegraphics[width=8cm]{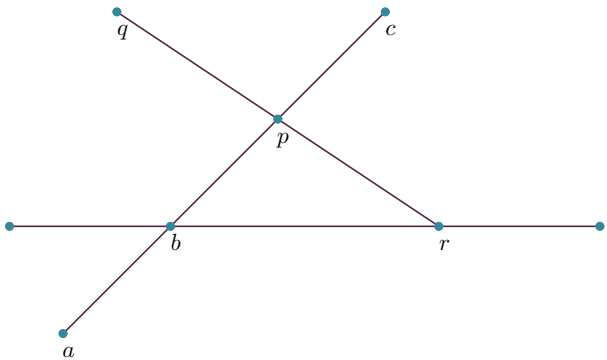}
	\caption{An example configuration of 8 points.
	Here $L(p) = \{\{p, q, r\}, \{a, b, c, p\}\}$,
	$N(p) = \{a, b, c, q, r\}$,
	$\mathrm{coll}(p) = 5$, and $\mathrm{deg}_G(p) = 5$.}
	\label{fcoll}
\end{figure}

\section{New Lower Bound}
\label{slow}

Before describing our algorithm, we present two lemmas as follows.

\begin{lem}
\label{lmis}
Any set $P$ of $n$ points contains two disjoint noncollinear subsets
$R$ and $T$ such that
$|R| + |T| \ge 2n \mathbin{/} ( \overline{\mathrm{coll}} (P) + 2)$.
\end{lem}
\begin{proof}
Let $G = (V, E)$ be the collinearity graph of $P$.
The vertex set of $G$ can be decomposed into two subsets $P_1$
and $P_2$ such that at least half of the edges of $G$
have one endpoint in each of these sets.
This can be done as follows.  Start with empty $P_1$ and $P_2$.
For every vertex $v$ of $G$ in some order, add it to
$P_1$ if $v$ has fewer neighbours in $P_1$ compared to $P_2$,
and add it to $P_2$, otherwise.
Let $P_1$ and $P_2$ be such a decomposition of $V$ and let $H$ be the graph
obtained from $G$ by removing all edges between $P_1$ and $P_2$.
Clearly, since the number of the edges of $H$ is at most half of that of $G$,
the average degree of $H$ is also at most half of the average degree of $G$.

Let $S$ be a maximal (that is, non-extensible) independent set in $H$ and
let $R = S \cap P_1$ and $T = S \cap P_2$.
Both $R$ and $T$ are independent sets in $H$.
Since no edge between the vertices of $P_1$ (and hence $R$)
is removed in $H$, $R$ is also an independent set in $G$.
Symmetrically, $T$ is also an independent set in $G$.

Turan's lower bound of $n \mathbin{/} ( \overline{d} + 1 )$ for the size of a maximal independent set in
a graph with $n$ vertices,
in which $\overline{d}$ is the average degree of the graph,
can be attained using the greedy algorithm that iteratively
selects vertices ordered increasingly by their degrees and
removes the selected vertex and its neighbours \cite{halldorsson97}.
Applying Turan's bound to $H$ yields that $\overline{d}$ is
at most half of the average degree of $G$.
Hence, we have $\overline{d} \le \overline{\mathrm{coll}} (P) \mathbin{/} 2$,
implying that $|S| \ge 2n / ( \overline{\mathrm{coll}} (P) + 2)$, as required.
\end{proof}

The lower bound for the greedy algorithm used in Lemma~\ref{lmis}
is not the best possible; it is actually a weaker form of
the celebrated Caro-Wei lower bound, which
has been improved by several authors (see, for instance, \cite{halldorsson97} and \cite{harant11}).
However, Turan's bound, which is tight for graphs consisting of
disjoint cliques, depends only on the average degree and
yields a cleaner bound for the size of
noncollinear sets in Lemma~\ref{lmis}.

\begin{lem}
\label{llow}
Let $R$ and $T$ be two noncollinear subsets of a set of points $P$.
Then, the points in $R \cup T$ are in general position.
\end{lem}
\begin{proof}
For each line with at least three points in $P$,
each of $R$ and $T$ can include at most one point.
Therefore, in their
union there are at most two points on each such line in $P$.
\end{proof}

We now present our algorithm for finding a subset $S$ of a
set of points $P$ in general position (Algorithm~\ref{agpss}) and
prove a lower bound for the size of the set it finds (Theorem~\ref{tlow}).

\begin{algorithm}[!ht]
    \caption{General position subset selection}
    \label{agpss}

    \SetKwInOut{Input}{Input}
    \SetKwInOut{Output}{Output}

    \Input{A set $P$ of points in the plane.}
    \Output{A subset $S$ of $P$ in general position.}
        Obtain all maximal subsets $L$ of collinear points in $P$.
        One way to do this is to move the points to the dual plane;
        three (or more) points are collinear if their corresponding lines
        in the dual plane intersect each other at a common point.
        Regular plane sweeping can identify these intersections.\

	Construct the collinearity graph $G$ from $L$ for the set of input points $P$.\

        Decompose the vertices of $G$ into two sets $P_1$ and $P_2$,
        such that there are at least $\mathrm{coll}(P) / 2$ edges with one
        end point in each of these sets, as described in Lemma~\ref{lmis}.
	Obtain graph $H$ from $G$ by removing all edges between $P_1$ and $P_2$.\

        Find a maximal independent set $S$ in $H$ as described in Lemma~\ref{lmis}.\

        Add each point in $P \setminus S$ to $S$, if it and no
        two points in $S$ are collinear.
\end{algorithm}
\begin{thm}
\label{tlow}
Any set $P$ of $n$ points in the plane contains a
subset in general position of size at least
$\max \{ 2 n^2 \mathbin{/} (\mathrm{coll} (P) + 2n) , \sqrt{s_{\textrm{GPSS}}} \}$.
\end{thm}
\begin{proof}
We use Algorithm~\ref{agpss} to prove this bound.
In the step 4 of the algorithm, $S$ has at
least $2n \mathbin{/} ( \overline{\mathrm{coll}} (P) + 2)$ points, as argued in Lemma~\ref{llow}.
Since $\mathrm{coll} (P) = n \cdot \overline{\mathrm{coll}} (P)$, the minimum size of $S$ can
be rewritten as $2n^2 \mathbin{/} ( \mathrm{coll} (P) + 2n)$.
In step 5, $S$ is made maximal.  By the same argument
mentioned in the introduction for the greedy GPSS algorithm,
it can be shown that after this step the size of $S$ is
at least $\sqrt{s_{\textrm{GPSS}}}$.
\end{proof}
To evaluate the improvement Theorem~\ref{tlow} yields compared to
Cao's algorithm, we experimentally compare them in the next section.

\section{Experimental Comparison of GPSS Algorithms}
\label{sgen}
To experimentally evaluate GPSS algorithms,
the first challenge is obtaining good test data
with a large number of collinear subsets of points.
To our knowledge, the available test data are mostly in
general position and thus are imperfect for evaluating GPSS algorithms.
Therefore, for our purpose it seems necessary to generate
point configurations with the desired rate of collinear points.
It is, however, very difficult to generate interesting configurations,
in which many points appear on each line and each point appears
on several lines, forming cycles with different lengths.
Indeed, it would be much easier to generate hypergraphs and
possibly convert them to point configurations.

A hypergraph can be created from a point configuration by
allocating a vertex in the hypergraph to each point and an
edge containing the vertices mapped to each maximal set
of collinear points.  Since lines can intersect in at
most one point, each pair of the edges in the resulting
hypergraph intersect in at most one vertex; hypergraphs
with this property are called \emph{linear}.
The other direction, i.e., obtaining a point configuration
from a linear hypergraph is not as trivial though, as the following
question states:

\begin{ques}
\label{qmap}
Given a linear hypergraph, is it possible to map its
vertices into points on the two-dimensional plane, such that three of
these points are collinear if and only if their corresponding vertices
are on the same edge of the hypergraph?
\end{ques}

There are hypergraphs for which the answer to the above question
is no; the Fano and Pappus configuration without the
Pappus line (the line containing the intersection points) are
two such examples.
It would be interesting to find the exact conditions for which
the answer to the above question is yes.
In other words, when do linear hypergraphs have a straight-line drawing?

The problem of drawing configurations is closely related to Question~\ref{qmap}:
a $k$-configuration with $n$ vertices, also denoted as
an $(n_k)$-configuration, is a linear hypergraph in which each
vertex is in exactly $k$ edges ($k$-regular) and each edge contains
exactly $k$ vertices ($k$-uniform) (for a detailed monograph on
configurations, the reader may consult \cite{grunbaum09} or \cite{pisanski12}).
Being defined more
than a century ago, $k$-configurations are one of the
oldest combinatorial structures, predating even the definition of
hypergraphs.
Both the identification of the combinatorial structure of $(n_k)$
configurations and their geometric realization have been the goal
of several mathematicians (Gropp discusses some of the history of
$k$-configurations \cite{gropp04} and their realization \cite{gropp97}).
Nevertheless, there is still very little known about the structure and the number
of geometric and even combinatorial $k$-configurations in general
(see for instance \cite{bokowski15}).

Even the existence of a fast algorithm for identifying
(or finding) geometric $k$-configurations may not help in answering
Question~\ref{qmap}; in a $k$-configuration, some collinear points
may not appear in an edge, while the mapping discussed before
Question~\ref{qmap} requires otherwise.

Question~\ref{qmap} also seems to be a generalization of the
famous stretchability problem for arrangements of pseudolines.
In the stretchability problem, one is given a simple arrangement of pseudolines (any
two pseudolines cross exactly once and no three pseudolines meet at one point), and
the question is whether the arrangement can be stretched, that is,
transformed into an equivalent arrangement of straight lines \cite{mnev88}.
For stretchability and similar problems, Schaefer introduced the
complexity class ER, for which these problems are complete \cite{schaefer09}.
In complexity-theoretic hierarchy, this class lies between NP and PSPACE.

The disappointing prospects of answering Question~\ref{qmap},
even when limited to $k$-configurations, suggest trying other
methods for generating interesting GPSS test data.

\subsection{Generating Random Hypergraphs}
\label{shyp}
The algorithms discussed in this paper do not directly use the
location of the points in the plane.
They extract point-line collinearity relations (or equivalently the underlying
linear hypergraph) as a first step;
therefore, the exact location of the points is not significant
in these algorithms.  Thus, instead of the locations of the points,
these algorithms can take the underlying hypergraph as input.

Although these algorithms expect a linear hypergraph
that can be drawn with straight lines, it would still be
interesting to evaluate their performance on general linear
hypergraphs, especially because of the difficulty of generating
configurations with complex point-line collinearity relations.
Thus, reminding that
the tests considered in this section target a more general problem
than GPSS, we generate random linear hypergraphs for testing GPSS
algorithms.

Our hypergraph construction is incremental, in contrast to Erd\H{o}s-R\'{e}nyi
random hypergraphs in which edges are either included or rejected with
some probability (see for instance \cite{dudek17}).
$H_d (n, m)$ is the hypergraph with $n$ vertices
and $m$ edges.  Initially all edges are empty.
The following step is repeated $dn$ times:
A vertex $v$ and an edge $e$ are chosen uniformly at random and $v$ is
added to $e$, if it does conflict with the linearity of the hypergraph.
Given that the above step is repeated $dn$ times, the average degree
is at most $d$.  The following probabilistic argument obtains a
lower bound for the expected average degree of the resulting hypergraph.

We denote with $Prob(R_i)$ the probability of rejection in the $i$-th step of the
algorithm, i.e., when adding vertex $v_i$ to edge $e_i$.
The rejection happens if i) $v_i$ already belongs to $e_i$, or
ii) $v_i$ is on an edge $e_j$,
one of whose vertices $v_k$ appears on $e_i$
(if $v_i$ is added to $e_j$, the hypergraph is no longer linear due to the intersections of $e_i$ and $e_j$).
Let $P_i (v \in e)$ be the probability of vertex $v$ being added to
edge $e$ in one of the first $i$-th steps.
\[
Prob(R_{i+1}) \le
P_i (v_i \in e_i) +
\sum_{j=1}^{m} \sum_{k=1}^{n}
P_i (v_i \in e_i)
P_i (v_k \in e_j)
P_i (v_k \in e_i)
\]

Given that $P_i (v \in e) \le \frac{i}{nm}$, the probability
of rejection can be simplified as follows.
\[
Prob(R_{i+1}) \le
\frac{i}{nm} + \frac{i^3}{nm^2}
\]

To obtain the expected value of the total number of rejections ($R$),
we add up $Prob(R_i)$ for all rounds of the algorithm (except
the first one, which is never rejected):
\[
Exp(R) \le
\sum_{i=2}^{dn} Prob(R_i) \le
\sum_{i=1}^{dn-1} \frac{i}{nm} + \sum_{i=1}^{dn-1} \frac{i^3}{nm^2}
\]

Simplifying this inequality, we obtain:
\[
Exp(R) < \frac{3d^2 n}{4m}
\]

Thus the expected average degree ($d_{\textrm{avg}}$) is:
\[
Exp(d_\textrm{avg}) \ge d - \frac{3 d^2} {4m}
\]

\subsection{Evaluating GPSS Algorithms}
We evaluate the following four GPSS algorithms.
\begin{description}
\item[Ind]
Algorithm~\ref{agpss}, which constructs the result by merging two
disjoint noncollinear subsets.
\item[Inc \cite{cao12}]
The greedy algorithm described in the Introduction,
which constructs a subset of input points by iteratively adding points
in some arbitrary order,
unless they are collinear with two other points in the set.
\item[Inc-Min]
Like Inc except that in each iteration the point with
the minimum number of collinear points (among the remaining points)
is selected.
\item[Dec]
Starts with all points in the set and iteratively removes
the points with the most number of collinearity relations with
other points in the set.
\end{description}

These algorithm are evaluated on two sets of test cases.  The first set
is obtained based on the method described in Section~\ref{shyp}.
The second set is generated by placing point on the two-dimensional
plane in some pattern, mostly grids.
The implementation of the algorithms and the test cases are
available at \url{https://github.com/aligrudi/gpss}.

Table~\ref{tsize} shows the average ratio of the size of the
set returned by any of the algorithms to the largest reported set.
The Inc-Min algorithm finds largest sets, outperforming
other algorithms by at least three percent on average.
Although the performance of Ind is close to Inc-Min,
there is an observable gap in their performance.

\begin{table}
\center
\caption{Average ratio of the size of the set returned by each algorithm
to the size of the best reported set}
\label{tsize}
\begin{tabular}{|l|c|}
\hline
Algorithm	& Average Ratio to the Best \\
\hline
Ind	& 96.6 \\
Inc	& 82.2 \\
Inc-Min	& 99.4 \\
Dec	& 87.4 \\
\hline
\end{tabular}
\end{table}

\begin{table}
\center
\caption{The ratio of the number of times each algorithm finds the largest set}
\label{tmax}
\begin{tabular}{|l|c|}
\hline
Algorithm	& Ratio of Best Performance \\
\hline
Ind	& 41.3 \\
Inc	& \ 9.5 \\
Inc-Min	& 92.1 \\
Dec	& 20.6 \\
\hline
\end{tabular}
\end{table}

Table~\ref{tmax} shows the ratio of the number of cases in
which each algorithm finds the largest set in general position;
this again suggests that Inc-Min obtains the largest sets most of the time.
These results suggest that it may be possible to obtain a better
lower bound for GPSS, based on Inc-Min.

\section{Concluding Remarks}
\label{scon}

GPSS can also be formulated in terms of Hypergraph strong
independence.  However, it is very surprising that the extensive
studies on independence number of hypergraphs are mostly focused on
weak independence (see \cite{eustis13} for a summary),
in which the independent set can include any
but not all of the vertices of each edge.

\section*{Acknowledgements}
The author would like to thank William Kocay
for the helpful discussion about geometric $k$-configurations.

\bibliographystyle{unsrt}

\end{document}